\documentclass[12pt]{amsart}
\usepackage{a4wide,palatino,epsfig,latexsym,
amsmath,amsthm,graphicx,hyperref,enumerate,morefloats,color, natbib, soul}

\usepackage{tikz}

\title{Irreversible evolution, obstacles in fitness landscapes and persistent drug resistance}

  \author{Kristina Crona}

\theoremstyle{plain}

\newtheorem{theorem}{Theorem}[section]

\newtheorem{conjecture}{Conjecture}
 
\theoremstyle{definition}

\newtheorem{definition}[theorem]{Definition}

\newtheorem{example}[theorem]{Example}

\newtheorem{observation}[theorem]{Observation}

\begin{document}

\maketitle

%Koll: Mkt empiri: wenreich 18, ej Harms Sailer 18, Ej jakub,

\begin{abstract}
We use fitness graphs, or directed cube graphs, for analyzing evolutionary
reversibility. The main application is antimicrobial drug resistance.
Reversible drug resistance has been observed both clinically and experimentally.
If drug resistance depends on a single point mutation, then a possible scenario
is that the mutation reverts back to the wild-type codon after the drug has been discontinued,
so that susceptibility is fully restored. In general, a drug pause does not automatically
imply fast elimination of drug resistance. Also if drug resistance is reversible, 
the threshold concentration for reverse evolution may be lower than for forward evolution.
For a theoretical understanding of evolutionary reversibility,
including threshold asymmetries, it is necessary to analyze
obstacles in fitness landscapes. 
We compare local and global obstacles,  obstacles for forward
and reverse evolution, and conjecture that favorable landscapes for forward evolution
 correlate with evolution being reversible. 
 Both suboptimal peaks and plateaus are analyzed with some observations
 on the impact of redundancy and dimensionality.
Our findings are compared with laboratory studies on irreversible malarial drug resistance.
\end{abstract}

\section{Introduction}
Penicillin was introduced on a large scale 1940 and the spread of penicillin resistance was
documented already in 1942 \citep{lp}. The development of antimicrobial 
drug resistance is an evolutionary
process, and so is the reverse
adaptation back to the
the drug-free environment.
Reverse evolution that restores
the original genotype has been 
observed for HIV patiens \citep{castro, yk},
and in experiments
for other pathogens \citep{bnb}.
However, expectations of fast reversal 
of antimicrobial resistance after a drug pause 
have not always been realized.
Failures to restore susceptibility includes
 nation-wide long term programs
 \citep{sga, esh}.

Costly drug resistance is not likely to persist
in a drug-free environment. If the original wild-type 
is available then regrowth 
may restore susceptibility
(no evolution is necessary).
In addition to extinction and reversion,
a  possible fate for a resistant
genotype is that new
mutations accumulate,
sometimes referred to as
compensatory mutations.
Such mutations decrease the
cost of resistance in the
drug free environment,
and there is usually 
an impact on susceptibility 
as well.

Evolution is described as genotypically irreversible
if the population cannot adapt back to
the original genotype,
and  phenotypically irreversible
if it cannot adapt back
to the original phenotype.
Because of genetic redundancy
in the sense that different sequences code
for the same phenotype, evolution
can be phenotypically reversible
 even if it is genotypically irreversible
 \citep{kjc}.
 An extensive laboratory study on
costly drug resistance  for
12 different antibiotics showed
partly successful adaptation
to the drug free environment
through compensatory 
mutations \citep{dunai}.
However, neither the original
genotype, nor the phenotype,
was restored for any of the drugs,
and in most cases the new genotypes
had clearly lower fitness than
the original wild-type.
In contrast, for some experiments of similar
type the original genotype was restored
in the majority of the trials \citep{bnb},
or at least in some proportion of the trials 
 \citep{mbl, nba}.
If the original genotype 
cannot be restored in experiments,
the reason could be that
evolution is  genotypically
irreversible. Another possible
explanation is that an abundance
of available compensatory mutations
make genotypic reversion unlikely. 
For more background, reversal of drug and
pesticide resistance is
reviewed in \cite{aeb}.

Here, the main topic is genotypically
irreversible evolution
and obstacles that cause
irreversibility. 
For a thorough analysis  it is useful to
consider fitness landscapes.
In brief, the fitness of a genotype  
is a measure of its expected contribution
to the next generation. 
A fitness landscape assigns a fitness value $w_g$,
 i.e., a non-negative number, to each genotype $g$.
Fitness can be thought of as  a 
height coordinate in the landscape.
The level of drug resistance approximates fitness for
a pathogen under  drug exposure.

Throughout the paper we consider fitness landscapes for
 biallelic $L$-locus systems.
For instance, if $L=2$ the genotypes are represented as
 $00$, $10$, $01$ and $11$, where $00$ 
 denotes the wild-type. 
According to conventional
assumptions, the evolutionary process
for a population can be represented as a walk 
in the landscape where each step increases
the height, i.e., the process consists
of a sequence of single point mutations 
$0\mapsto1$ or $1 \mapsto 0$ such that each
mutation increases fitness. 
Unless otherwise stated,
no two genotypes have
the same fitness.
A genotype $g$ is defined as a peak if
all its mutational  neighbors
(genotypes that differ from $g$ at a single locus)
have lower fitness than $g$.

For an overview of evolutionary
potential  it is  convenient to use fitness graphs (Figure 1).
\begin{figure}
A)
\begin{tikzpicture}
[thick, color=black,scale=0.9]
%[scale=0.6, auto=left,every node/.style={circle, draw,
%thick,outer sep=5pt}]
  \node (n1) at (4,0) {00};
  \node  (n2) at (2,2)  {10};
  \node   (n3) at (6,2)  { 01};
  \node [color=red] (n4) at (4,4) {\bf 11};
  \foreach \from/\to in {n2/n4,n3/n4,n1/n2,n1/n3}
  \draw[very thick] [-> ](\from) -- (\to);
\end{tikzpicture}
 \quad
 B)
\begin{tikzpicture}
[thick, color=black,scale=0.9]
 \node [color=red] (n1) at (4,0) {\bf 00};
  \node  (n2) at (2,2)  { 10};
  \node  (n3) at (6,2)  { 01};
  \node  (n4) at (4,4) {11};
  \foreach \from/\to in {n4/n3,n4/n2,n2/n1,n3/n1}
  \draw[very thick] [-> ](\from) -- (\to);
  \end{tikzpicture}
\caption{The fitness graph  for the new environment (A) is favorable since
mutations can accumulate in any order, i.e., both trajectories from
the wild-type $00$ to $11$ are accessible. Reverse evolution from $11$
to $00$ is also straight forward (B).}
\end{figure}
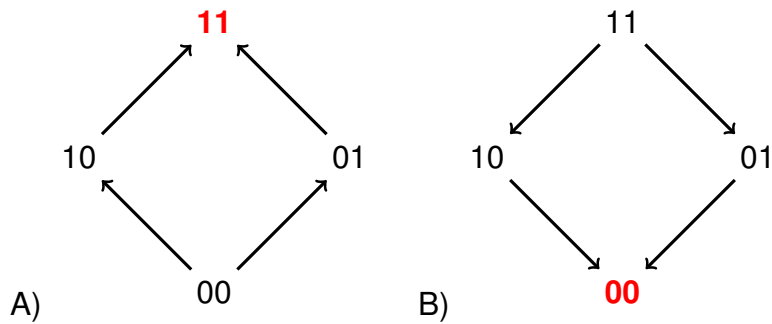
A fitness graph is a directed $L$-cube graph
such that each edge is directed toward
the genotype of higher fitness. 
A path in the graph that respects the
arrows is referred to as an accessible
evolutionary path.

Both forward evolution from the wild-type $00$ to $11$ and reverse evolution
from $11$ to $00$ are straight forward for the fitness graphs
in Figure 1 (the peaks are marked red). 
Mutations can accumulate in any order in the new environment,
and the same is true back in the original environment.
The graph 2A is less favorable than 1A, since only one
trajectory is accessible from $00$ to $11$,
and 2B has no accessible trajectory from $00$ to $11$.
For general $L$, the most favorable fitness graph
is similar to Figure 1. Informally the  
landscape is represented by an "all arrows up" graph.

The obstacles displayed in Figure 2 depend on epistasis, 
or gene interactions. The three fitness graphs have sign epistasis, i.e., the sign of the effect of a mutation, 
whether positive or negative, depends on background. 
The graph 2B, characterized by two peaks,
is said to have reciprocal sign epistasis.
For Graph 2B and C every other genotype is a peak
and the same construction works for any $L$  \citep{h}.
The graphs with 50 $\%$ peak density are called Haldane graphs (see also \citet{cks}).

Note that in the absence of sign epistasis the fitness graph can always 
be described as an all arrows up graph.
For more background, sign epistasis
was introduced in \citet{wwc} and early work on sign epistasis,
fitness graphs and related rank order based concepts includes \citet{pkw, de, ptk, cgb},
see also \citet{c14}. A main topic
concerns the relation between local properties
(such as reciprocal sign epistasis) and global properties
(such as peaks in the global fitness landscapes),
with recent progress in \citet{rpp, skk}.

\begin{figure}
A)
\begin{tikzpicture}
[thick, color=black,scale=0.9]
%[scale=0.6, auto=left,every node/.style={circle, draw,
%thick,outer sep=5pt}]
  \node (n1) at (4,0) {00};
  \node  (n2) at (2,2)  {10};
  \node   (n3) at (6,2)  { 01};
  \node [color=red] (n4) at (4,4) {\bf 11};
  \foreach \from/\to in {n2/n4,n3/n4,n1/n2,n3/n1}
  \draw[very thick] [-> ](\from) -- (\to);
\end{tikzpicture}
B)
\begin{tikzpicture}
[thick, color=black,scale=0.9]
%[scale=0.6, auto=left,every node/.style={circle, draw,
%thick,outer sep=5pt}]
  \node [color=red]  (n1) at (4,0) {\bf 00};
  \node  (n2) at (2,2)  {10};
  \node   (n3) at (6,2)  { 01};
  \node [color=red] (n4) at (4,4) {\bf 11};
  \foreach \from/\to in {n2/n4,n3/n4,n2/n1,n3/n1}
  \draw[very thick] [-> ](\from) -- (\to);
\end{tikzpicture}
 C)
 \begin{tikzpicture}
[very thick, color=black,->,scale=1.4]
\node (n0) at (0,0) {000};
\node  [color=red] (n1) at (1.2,1) {\bf 001};
\node  [color=red] (n2) at (0,1) {\bf 010};
\node (n3) at (1.2,2) {011};
\node  [color=red] (n4) at (-1.2,1) {\bf 100};
\node (n5) at (0,2) {101};
\node (n6) at (-1.2,2) {110};
\node  [color=red]  (n7) at (0,3) {\bf 111};
\draw (n0) -- (n1);
\draw (n0) -- (n2);
\draw (n0) -- (n4);
\draw (n3) -- (n1);
\draw (n3) -- (n2);
\draw (n3) -- (n7);
\draw (n5) -- (n1);
\draw (n5) -- (n4);
\draw (n5) -- (n7);
\draw (n6) -- (n2);
\draw (n6) -- (n4);
\draw (n6) -- (n7);
\end{tikzpicture}
\caption{Three graphs with sign epistasis. 
The graph B and all two locus subsystems  of C
have reciprocal sign epistasis.
Both B and C are Haldane graphs, i.e.,  
they represent fitness landscapes such that 
50 percent of the genotypes are peaks.}
\end{figure}
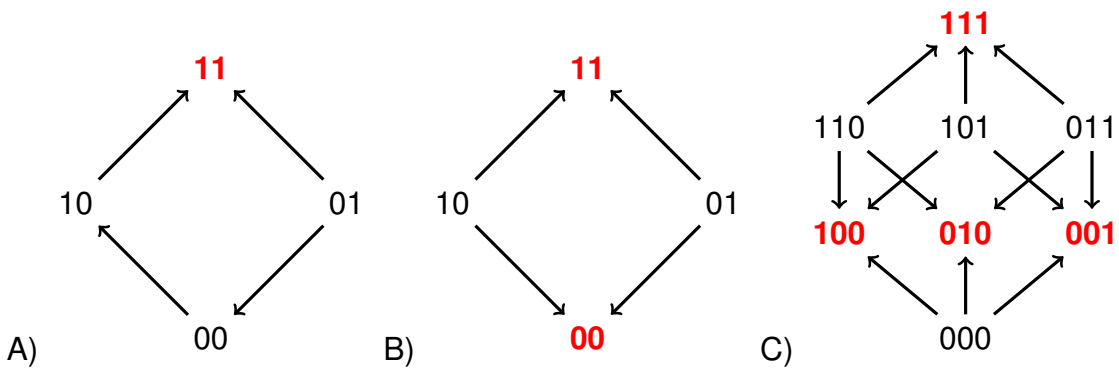
All arrows up graph and Haldane graphs are, in a sense,
opposite extremes.
For the sake of completeness, in addition to all arrows up graph
there a second - perhaps more exotic -
type of fitness landscapes that implies 
straight forward evolution. 
If one assumes that the  fitness
is capped at some number $M$ in a landscape, and that
all genotypes have several neighbors
with fitness $M$, then adaptation
is not difficult. From any starting point
 there exists a single point mutation 
 that results in maximal fitness. 
A slightly more general concept is convenient.
We define a fitness
landscape as a \emph{hop-to-top landscape}
if fitness is capped at some number $M$, and
if for  (almost) all genotypes, there is a short
accessible path to a genotype of  fitness $M (1- \epsilon )$ , 
where $\epsilon$ is some small number.

For instance, a hop-to-top landscape
can be constructed
by drawing fitness values
from a uniform distribution between $0$ and $M$,
and assign fitness randomly to genotypes.
For $L=10,000$ a genotype
is then expected to have 100 neighbors
with fitness at least $0.99M$.
%Adaptation from environment A to B,
 %and back to A again, should be 
 %straight forward under the
 %assumptions described, but the 
 %original wild-type is not likely to be 
 %restored because of  the rich
 %supply of alternative paths.
 Under the same assumption for the landscape
except that $L$ is slightly smaller, one  
still gets a hop-to-top landscape, whereas
a sufficiently low $L$-value results in an
unfavorable landscape
(both claims are easy to verify).
It has been proposed that similar
constructions (for large $L$) 
are relevant for speciation \citep{gav},
see also the result section.
 
Haldane graphs, all arrows up graphs
and hop-to-top landscapes are 
theoretical constructions that can
be used as a starting point for 
discussing obstacles in fitness landscapes
and evolutionary reversibility. However, 
few empirical fitness landscapes 
have proved to belong to the extremes.
%The impact of obstacles is less obvious
%for general fitness landscapes.
Obstacles have to be considered
in more general settings.

From a fitness landscape for antimicrobial
drug resistance one can determine
if evolution is reversible.
Whether or not reversion is plausible 
depends on other factors as well, 
including population size and 
mutation frequency  \citep{poh,mbl}.
%Theory on reverse evolution is immediately 
%applicable to antimicrobial drug resistance.
Such factors will not be discussed here.
Neither will we discuss more
elaborate methods for restoring the original wild-type
that depends on sequences of drugs \citep{mcg, gmc, ty}.

\section{results}
%Irreversible evolution for a two-locus systems can be represented 
%by the fitness graph 2B (from the mutant $11$ there are no accessible paths to $00$).
A case of Irreversible malarial drug resistance 
was identified in the study \citet{oh}. 
Several drug concentrations were considered in the study.
Figure 3 shows the fitness graph for the highest drug concentration,
and Figure 4 for the drug-free environment.
The genotype 1111 is the global peak and 0000 has the lowest fitness
in Figure 3, whereas  0000 is the global peak  in Figure 4.
As clear from Figure 4, there is no accessible path from $1111$
to $0000$. Consequently evolution is 
irreversible. (It is of course theoretically possible that a longer 
accessible path from $1111$ to $0000$ exists that
includes new mutations in addition to the reversions.)
We will return to the example repeatedly throughout the paper.

Section 2.1 and 2.2 analyze suboptimal peaks
and plateaus in fitness landscapes,
Sections  2.3  discusses reversibility, and Section 2.4 reversibility
and the impact of fluctuating drug concentrations.

\begin{figure}\label{small}
\begin{tikzpicture}
[very thick,black,->,outer sep=1mm, scale=0.8]
\node (n0) at (0, -4.0) {0000};
\node (n1) at (4.0, -2.0) {0001};
\node (n2) at (1.3333333333333335, -2.0) {0010};
\node (n3) at (5.0, 0) {0011};
\node [color=red] (n4) at (-1.333333333333333, -2.0){\bf 0100};
\node (n5) at (3.0, 0) {0101};
\node  (n6) at (1.0, 0) {0110};
\node  (n7) at (4.0, 2.0) {0111};
\node (n8) at (-3.9999999999999996, -2.0) {1000};
\node (n9) at (-1.0, 0) { 1001};
\node  (n10) at (-3.0, 0) {1010};
\node (n11) at (1.3333333333333335, 2.0) {1011};
\node (n12) at (-5.0, 0) {1100};
\node (n13) at (-1.333333333333333, 2.0) {1101};
\node  (n14) at (-3.9999999999999996, 2.0) {1110};
\node [color=red] (n15) at (0, 4.0) {\bf1111};
\draw (n0) -- (n1);
\draw (n0) -- (n2);
\draw (n0) -- (n4);
\draw (n0) -- (n8);
\draw (n1) -- (n3);
\draw  (n1) -- (n5);
\draw (n1) -- (n9);
\draw (n2) -- (n3);
\draw  (n2) -- (n6);
\draw  (n2) -- (n10);
\draw (n5) -- (n4);
\draw  (n6) -- (n4);
\draw (n12) -- (n4);
\draw (n3) -- (n7);
\draw (n5) -- (n7);
\draw (n6) -- (n7);
\draw (n9) -- (n8);
\draw (n8) -- (n10);
\draw (n12) -- (n8);
\draw (n11) -- (n3);
\draw (n11) -- (n9);
\draw (n11) -- (n10);
\draw  (n13) -- (n5);
\draw  (n13) -- (n9);
\draw (n13) -- (n12);
\draw (n6) -- (n14);
\draw (n10) -- (n14);
\draw (n12) -- (n14);
\draw (n7) -- (n15);
\draw (n11) -- (n15);
\draw (n13) -- (n15);
\draw (n14) -- (n15);
\end{tikzpicture}
\caption{The 16 genotypes represent all combinations of four mutations that individually increase
malarial drug resistance. 
The genotype 1111 is the global peak and the wild-type 0000 has the lowest fitness
in the drug environment. There are several accessible paths from 0000 to 1111.}
\begin{tikzpicture}
[very thick,black,->,outer sep=1mm, scale=0.8]
\node  [color=red] (n0) at (0, -4.0) {\bf 0000};
\node (n1) at (4.0, -2.0) {0001};
\node (n2) at (1.3333333333333335, -2.0) {0010};
\node (n3) at (5.0, 0) {0011};
\node (n4) at (-1.333333333333333, -2.0) {0100};
\node (n5) at (3.0, 0) {0101};
\node  (n6) at (1.0, 0) {0110};
\node  (n7) at (4.0, 2.0) {0111};
\node (n8) at (-3.9999999999999996, -2.0) {1000};
\node (n9) at (-1.0, 0) { 1001};
\node  (n10) at (-3.0, 0) {1010};
\node (n11) at (1.3333333333333335, 2.0) {1011};
\node (n12) at (-5.0, 0) {1100};
\node [color=red]  (n13) at (-1.333333333333333, 2.0) {\bf1101};
\node [color=red] (n14) at (-3.9999999999999996, 2.0) {\bf 1110};
\node (n15) at (0, 4.0) {1111};
\draw (n1) -- (n0);
\draw (n2) -- (n0);
\draw (n4) -- (n0);
\draw (n8) -- (n0);
\draw (n3) -- (n1);
\draw  (n5) -- (n1);
\draw (n1) -- (n9);
\draw (n3) -- (n2);
\draw  (n2) -- (n6);
\draw  (n2) -- (n10);
\draw (n5) -- (n4);
\draw   (n4) -- (n6);
\draw (n4) -- (n12);
\draw  (n3) -- (n7);
\draw (n5) -- (n7);
\draw (n7) -- (n6);
\draw  (n8) -- (n9);
\draw (n10) -- (n8);
\draw  (n12) -- (n8);
\draw (n3) -- (n11);
\draw (n11) -- (n9);
\draw (n11) -- (n10);
\draw  (n5) -- (n13);
\draw  (n9) -- (n13);
\draw (n12) -- (n13);
\draw (n6) -- (n14);
\draw  (n10) -- (n14);
\draw (n12) -- (n14);
\draw (n15) -- (n7);
\draw (n11) -- (n15);
\draw  (n15) -- (n13);
\draw (n15) -- (n14);
\end{tikzpicture}
\caption{The genotype 0000 is the global peak, whereas 1111 has low fitness in
the drug-free environment.
Evolution is irreversible since there is no accessible path from the genotype $1111$ to $0000$}
\end{figure}

\subsection{Suboptimal peaks}
The most simple example of a suboptimal
peak arises if  a double mutant with higher
fitness than the wild-type 
combines two detrimental
single mutations (Figure 2B).
For general $L$ any two-locus subsystem with reciprocal
sign epistasis  constitutes a global obstacle if independent of background.
For instance, assume that
\[
w_{00s} > w_{11 s}> w_{10 s}, w_{01s}    \quad \ast
\]
for all  $s$ of length $L-2$.
Then it is clear that some genotype of the form $g=11\tilde{s}$ is a suboptimal peak.
(Indeed, if $11\tilde{s}$ has maximal fitness among all genotypes of the form
 $11s$, then $11\tilde{s}$ is a peak with lower fitness than $00 \tilde s$.)

A variant of the same theme (bad+bad=good),
is that  the combined effect of replacing two blocks (sets) of  loci 
is positive, whereas the replacement of each block alone is negative. 
Such a system is sometimes referred to as lock-key system \citep{dds}
Similar to the condition $\ast$,  local obstacles constitute global obstacles
if independent of background.

\begin{definition}
For a block of length $L'<L$, assume that $\prec$ is an order of the genotypes in the $L'$-locus subsystem. 
The block is  {\emph{rank order preserving}} if 
 the following condition holds: 
\[
w_{g s} > w_{g's}   \text{ if }  g \succ g',
\]
for $g, g'$ in the $L'$-locus subsystem.
\end{definition}

The following observation is  immediate (and analogous to the implications of $\ast$).

\begin{observation}\label{global}
For each peak in a rank order preserving block
there is a corresponding peak in the global $L$-locus system.
\end{observation}

Figure 5 illustrates  Observation \ref{global}. The two loci in the middles constitute
a rank order preserving block. The four subsystems of the form
\[
\ast 00 \star,  \ast10 \star,  \ast 01 \star, \ast 11 \star
 \]
has reciprocal sign epistasis (marked by blue arrows).
Observation \ref{global} implies that the two peaks $\ast 00 \star,  \ast 11 \star$ in the rank order preserving block correspond
to two peaks in the global system.  The peaks in the global system are $0000$ and $0110$.

The key property of rank order preserving blocks holds in a more general setting.

\begin{definition}
For a block of length $L'<L$, assume that $\prec$ is an order of the genotypes in the $L'$-locus subsystem.
A block of length $L' < L$ is a {\emph{graph preserving block}} if for any two mutational neighbors $g$ and $g'$ in the $L'$-locus subsystem,
\[
w_{g s} > w_{g's}   \text{ if } g  \succ g' 
\]
\end{definition}
Note that the condition  implies that the fitness graph for the  $L'$-locus subsystems defined by the graph preserving block are independent of background (see the four marked 
subgraphs in Figure 4).

\begin{observation}\label{global2}
For each peak in a graph preserving subsystem,
there is a corresponding peak in the global $L$-locus system.
\end{observation}

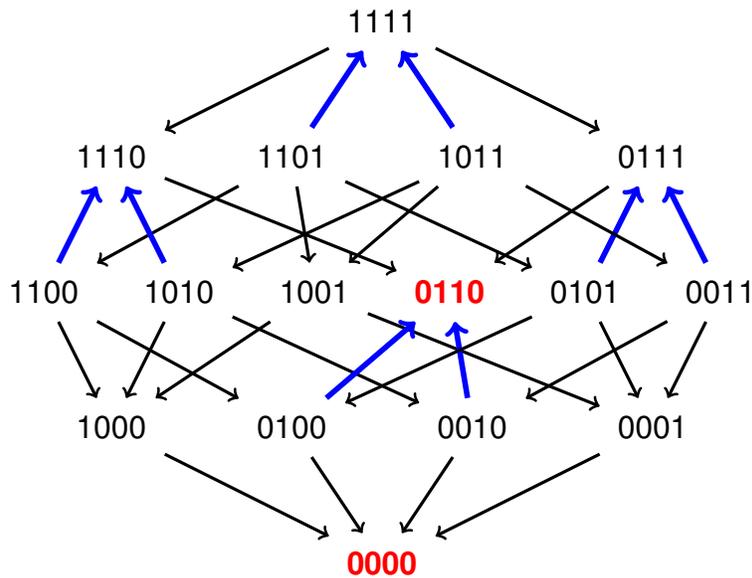
\begin{figure}\label{irr}
\begin{tikzpicture}
[very thick,black,->,outer sep=1mm, scale=0.9]
\node  [color=red] (n0) at (0, -4.0) {\bf 0000};
\node (n1) at (4.0, -2.0) {0001};
\node (n2) at (1.3333333333333335, -2.0) {0010};
\node (n3) at (5.0, 0) {0011};
\node (n4) at (-1.333333333333333, -2.0) {0100};
\node (n5) at (3.0, 0) {0101};
\node [color=red] (n6) at (1.0, 0) {\bf 0110};
\node  (n7) at (4.0, 2.0) {0111};
\node (n8) at (-3.9999999999999996, -2.0) {1000};
\node (n9) at (-1.0, 0) { 1001};
\node  (n10) at (-3.0, 0) {1010};
\node (n11) at (1.3333333333333335, 2.0) {1011};
\node (n12) at (-5.0, 0) {1100};
\node (n13) at (-1.333333333333333, 2.0) {1101};
\node (n14) at (-3.9999999999999996, 2.0) {1110};
\node (n15) at (0, 4.0) {1111};
\draw (n1) -- (n0);
\draw (n2) -- (n0);
\draw (n4) -- (n0);
\draw (n8) -- (n0);
\draw (n3) -- (n1);
\draw  (n5) -- (n1);
\draw (n9) -- (n1);
\draw (n3) -- (n2);
\draw [color=blue, line width=2.1 pt] (n2) -- (n6);
\draw (n10) -- (n2);
\draw (n5) -- (n4);
\draw [color=blue, line width=2.1 pt] (n4) -- (n6);
\draw (n12) -- (n4);
\draw [color=blue, line width=2.1 pt]  (n3) -- (n7);
\draw [color=blue, line width=2.1 pt] (n5) -- (n7);
\draw  (n7) -- (n6);
\draw (n9) -- (n8);
\draw (n10) -- (n8);
\draw (n12) -- (n8);
\draw (n11) -- (n3);
\draw (n11) -- (n9);
\draw (n11) -- (n10);
\draw (n13) -- (n5);
\draw (n13) -- (n9);
\draw (n13) -- (n12);
\draw (n14) -- (n6);
\draw  [color=blue, line width=2.1 pt] (n10) -- (n14);
\draw  [color=blue, line width=2.1 pt]  (n12) -- (n14);
\draw (n15) -- (n7);
\draw   [color=blue, line width=2.1 pt] (n11) -- (n15);
\draw  [color=blue, line width=2.1 pt] (n13) -- (n15);
\draw (n15) -- (n14);
\end{tikzpicture}
\caption{The subsystems determined by the central block of length 2, marked with blue arrows, have reciprocal sign epistasis on all backgrounds.
The obstacle prevents evolution from 1111 to the global peak 0000. All arrows except the blue
ones  point toward 0000.}
\end{figure}

A closer look at the study of malarial drug resistance for the drug-free environment (Figure 4) reveals a pattern 
that is very similar to the graph preserving block shown in Figure 5. For an easy comparison, Figure 6
is a copy of Figure 4 with the relevant arrows marked blue. The arrows agree with Figure 5 except
for a single arrow marked red. The red arrow leads directly to a suboptimal peak (the arrow could otherwise
have served as an escape). It is "almost true" that a graph preserving block prevents reverse evolution.

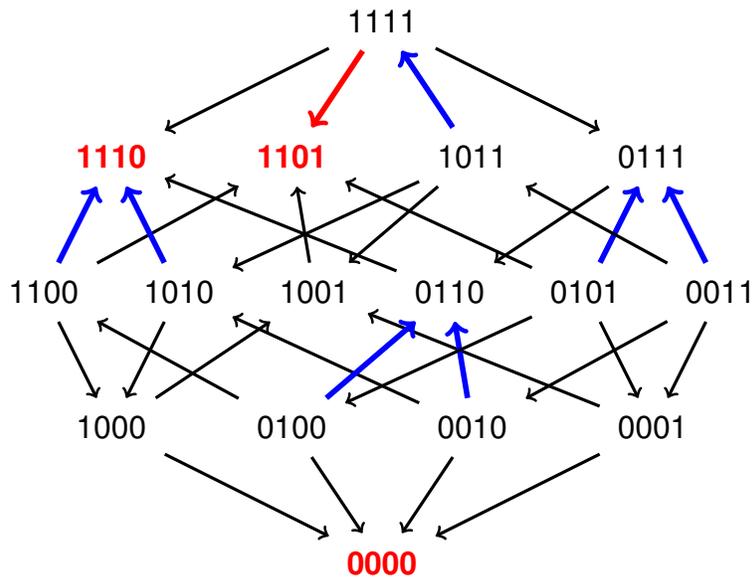
\begin{figure}

\begin{tikzpicture}
[very thick,black,->,outer sep=1mm, scale=0.9]
\node  [color=red] (n0) at (0, -4.0) {\bf 0000};
\node (n1) at (4.0, -2.0) {0001};
\node (n2) at (1.3333333333333335, -2.0) {0010};
\node (n3) at (5.0, 0) {0011};
\node (n4) at (-1.333333333333333, -2.0) {0100};
\node (n5) at (3.0, 0) {0101};
\node  (n6) at (1.0, 0) {0110};
\node  (n7) at (4.0, 2.0) {0111};
\node (n8) at (-3.9999999999999996, -2.0) {1000};
\node (n9) at (-1.0, 0) { 1001};
\node  (n10) at (-3.0, 0) {1010};
\node (n11) at (1.3333333333333335, 2.0) {1011};
\node (n12) at (-5.0, 0) {1100};
\node [color=red]  (n13) at (-1.333333333333333, 2.0) {\bf1101};
\node [color=red] (n14) at (-3.9999999999999996, 2.0) {\bf 1110};
\node (n15) at (0, 4.0) {1111};
\draw (n1) -- (n0);
\draw (n2) -- (n0);
\draw (n4) -- (n0);
\draw (n8) -- (n0);
\draw (n3) -- (n1);
\draw  (n5) -- (n1);
\draw (n1) -- (n9);
\draw (n3) -- (n2);
\draw  [color=blue, line width=2.1 pt] (n2) -- (n6);
\draw  (n2) -- (n10);
\draw (n5) -- (n4);
\draw  [color=blue, line width=2.1 pt]  (n4) -- (n6);
\draw (n4) -- (n12);
\draw [color=blue, line width=2.1 pt] (n3) -- (n7);
\draw [color=blue, line width=2.1 pt] (n5) -- (n7);
\draw (n7) -- (n6);
\draw  (n8) -- (n9);
\draw (n10) -- (n8);
\draw  (n12) -- (n8);
\draw (n3) -- (n11);
\draw (n11) -- (n9);
\draw (n11) -- (n10);
\draw  (n5) -- (n13);
\draw  (n9) -- (n13);
\draw (n12) -- (n13);
\draw (n6) -- (n14);
\draw [color=blue, line width=2.1 pt] (n10) -- (n14);
\draw [color=blue, line width=2.1 pt](n12) -- (n14);
\draw (n15) -- (n7);
\draw [color=blue, line width=2.1 pt](n11) -- (n15);
\draw  [color=red,  line width=2.1 pt] (n15) -- (n13);
\draw (n15) -- (n14);
\end{tikzpicture}
\caption{The fitness graph agrees with Figure 4. Similar to Figure 5 all arrows  right  below genotypes of the form $\ast 11\star$ (marked blue) point up, with one
exception (marked red). The red arrow leads directly to a suboptimal peak.}
\end{figure}

As demonstrated, the existence of single graph preserving block with suboptimal
peaks implies that there are suboptimal peaks in the global fitness landscape. 
The following schematic example illustrates
the impact of multiple rank order preserving blocks.

\begin{example}\label{unfav}
Assume that  the genotypes in an $L$-locus system can be partitioned into
blocks consisting of two loci, where for each block (using informal notation)
$w_{11}>w_{00}>w_{10}> w_{01}$.
For $L=6$ there are eight peaks:
\[
000000,
110000,
001100,
000011,
111100,
110011,
001111,
111111
\]
Evolution from $000000$ to $111111$  requires passing three obstacles, i.e.,
moving from $00$ to $11$ for each one of the three blocks.
Analogously, for $L=40$ there are about a million peaks, a billion genotypes,
and 20 obstacles.
In general, the peak density  $2^{-L/2}$ decreases by $L$.
However, it is fair to say that the landscapes are equally (un)favorable for all $L$
since the number of obstacles ($L/2$) is proportional to $L$.
\end{example}

\begin{observation}
If the peak density decreases by $L$ for a class of fitness landscapes,
it does not follow that the landscape becomes more favorable by $L$.
\end{observation}

\begin{observation}\label{product}
If the $L$ sequence (the genome) can be partitioned into graph preserving blocks,
then the number of peaks equals the product of the number of peaks in each block.
\end{observation}

\begin{proof}
Let $b_1, \dots, b_r$ be the blocks and assume that $b_i$ has $n_i$ peaks.
Let $g=g_1 \ldots g_r$  be a genotype such that $g_i \in b_i$.
Then $g$ is a peak if and only if each $g_i$ is a peak in the block $b_i$.
Consequently, there are in total  $n_1 \times \dots \times n_r$  
peaks in the global fitness landscape. 
\end{proof}

Consider the category of fitness landscapes such that the genome
can be partitioned into rank order preserving blocks.
The block landscapes introduced in  \citet{pm} belong to the category.
Specifically, the landscapes are defined so that each block contributes independently to fitness, 
and fitness values within blocks are assigned randomly. 
Observation \ref{product} for block landscapes was stated in \citet{sk}. 
Landscapes in the category are similar in that obstacles in each block
have a global impact (in contrast to for instance hop-to-top landscapes where
subsystems with reciprocal sign epistasis have no relevance).
For landscapes in the category, the problem of finding the global peak in the $L$-locus system 
is equivalent to the combined problem of finding 
the optimal sequence for each block (as in Example \ref{unfav}).
It follows that adding blocks, all else equal, does not make the fitness landscapes 
more favorable.

Fitness landscapes such that the $L$-sequence
can be partitioned into graph order preserving blocks 
differ substantially from the rank order preserving case.
The reason is that the optimal sequence for a particular
 block may depend on background. An example is the
 following 4-locus system.

\begin{example}\label{plea}
Assume that  the genotypes in an $4$-locus system can be partitioned into
blocks consisting of two loci, where $00$ and $11$ have higher
fitness than the intermediates $10$ and $01$ in each block (again using informal
notation). The peaks are
\[
0000,
1100,
0011,
1111,
\]
Moreover, assume that 
\[
w_{1111}>w_{0000}> w_{1100}>w_{0011}
\]
Evolution from $0000$ to $1111$  is difficult since
both $0000 \mapsto 1100$  and $0000 \mapsto 0011$  
decrease fitness. 
\end{example}

%The optimal sequence for each block in  Example \ref{plea} depends 
%on background, in contrast to Example \ref{unfav}.

It is instructive to compare $\ast$ and other rank order conditions discussed
here with conventional models of fitness landscapes.
The condition $\ast$  does obviously not hold in the
absences of sign epistasis, in particular not for additive fitness landscapes.  
The condition  $\ast$  is also incompatible with 
 hop-to-top landscapes constructed  by
 a random-fitness assignment
 (see the introduction).
 The reason is that for a  random fitness landscape $w$,
the inequality $ w_{00s} > w_{10 s}$ cannot hold for all 
$s$ if $L$ is large.

An empirical study uses assumptions similar to the hop-to-top
landscapes, i.e., fitness is randomly assigned from a uniform distribution between $0$ and $1$,
for describing worst-case scenarios for adaptation  \citep{gla},
with the difference that fitness is assigned to phenotypes rather than genotypes.
Note that $\ast$ cannot hold for a random fitness landscape $w$
as described. The reason is that genotypes of the form $00s$ correspond to
many phenotypes, and similarly for genotypes $10s$.
Consequently there are  both $s'$ such that $w_{10s'}> w_{00s'}$ and 
$s''$ such that $w_{10s''}< w_{00s''}$.
In other words, the assumptions on $w$ are not compatible with $\ast$
or similar rank order conditions.
It follows  that random assumptions do not describe worst-case scenarios
for fitness landscapes in settings where rank order preserving blocks
are important.

\subsection{Suboptimal plateaus}
Some fitness landscapes have a high degree of redundancy.
Consequently, it is of interest to consider
landscapes where mutational neighbors are allowed to have
the same fitness. Fitness graphs for such landscapes can be drawn
similarly to  standard fitness graphs, except that some arrows
would be replaced by edges.

\begin{definition}
A genotype $g$ belongs to a suboptimal plateau in
a fitness landscape if
\begin{itemize}
\item[(i)]  All neighbors have the same or lower fitness than $g$, and
\item[(ii)] at least one genotype in the fitness landscape have higher fitness than  $g$.
\end{itemize}
\end{definition}

For fitness landscapes with a high degree of redundancy,
an evolving population may be unable to reach a genotype
of high fitness because of suboptimal plateaus. 
The following example shows the impact of
plateaus.

\begin{example}
Assume that an 8-locus system consists of two blocks of 4 loci.
The first block is in state $0$ for the following eight sequences:
\[
0000, 1000, 0001, 1100, 0110, 1001, 1110, 1101,
\]
and in state $1$ for the remaining eight sequences
\[
0100, 0010, 1010, 0101, 0011,  1011, 0111, 1111.
\]
Notice that for any sequence, a single mutation can change the 
state of the block (from 0 to 1, or from 1 to 0).
Assume that the second block has similar properties.
Then the $8$-locus system  has four states $00, 10, 01, 11$
determined by the state of each block.
If the  fitness of a genotype in the 8-locus system is determined its state,
then the landscape is analogous to  a biallelic 2-locus system.

For instance, if $w_{11}>w_{00}>w_{10}>w_{01}$, where $w_{ij}$ denotes
the fitness for the state $ij$, then the 64 genotypes that represent the
state $00$ constitute a suboptimal plateau.  Neutral mutations are 
available, but no sequence of neutral mutations result in a genotype
such that beneficial mutations are possible.
 \end{example}

By using a similar construction, one can obtain a
landscape with arbitrary redundancy from a 
fitness landscape $w$ without redundancy.
Specifically, assume that $w$ has $s$ peaks and that no two genotypes
have the same fitness. Construct $s$ blocks of $r$ loci,
such that 50 percent of the sequences in each block has state 0
and 50 percent state 1, and such that for any sequence 
a single mutation can change the state. 
(In the previous example $r=4$ and $s=2$. The construction is not more difficult for larger
$r$-values.)
If $r_1, \dots, r_s$ represent the
states for a genotype in the $L=rs$-locus system, then one assigns the
fitness  $w_{r_1 \dots r_s}$ to the genotype. The observation below follows.

\begin{observation}
For every fitness landscape with no redundancy,
one can construct an landscape with an arbitrarily high 
degree of redundancy,
such that each suboptimal peak in the first landscape
corresponds to a suboptimal plateau in the second landscape.
\end{observation}

\subsection{Irreversible evolution} %for $L=2$}
If the optimal genotypes for an organism
differ between two environments A and B,
it is interesting to analyze forward and reverse evolution.
The motivating example is costly drug resistance.
This section does not include empirical examples,
but rather an analysis of small systems from
a theoretical point of view.

If the adaptation to a new environment depends on a single point mutation,
then evolution is reversible. However, the case $L=2$ is already 
more interesting.
Assume that $00$ is optimal in original environment
and $11$ in the new environment, so that
in a limited sense there
is a trade-off between optimal fitness
in the two environments.
Then there are  (in principe) two fitness graphs that allow
for forward evolution 
described  by the graphs 1A (all arrows up) and 2A  (exactly one arrow down).

Fitness graph 1A for forward evolution:  
Both paths $00 \mapsto 10 \mapsto 11$ and $00 \mapsto 01 \mapsto 11$
are accessible, i.e., the mutations are beneficial independent of background.
In this case it seems plausible that reverse mutations would be beneficial in
the original environment (see also the conjecture below).

Fitness graph 2A for forward evolution:  
By assumption, exactly one path is accessible,
described as $00 \mapsto 10 \mapsto 11$.
The mutation $0 \mapsto 1$ at the right locus
is only beneficial if the left locus is mutated, i.e.,
there is no independent advantage for the mutation in the new environment. 
Consequently, the advantage may have nothing to do
with the new environment but rather constitute an adjustment because
of the left substitution. If that is the case, it seems likely that $11$ 
has higher fitness than $10$ {\emph{also in the original environment}}, which
would imply irreversible evolution (the fitness graph would agree with Figure 2B).

Based on the discussion, one can try to relate 
forward and reverse evolution for $L=2$.
The assumptions are that
11 has highest fitness in the new environment, 00 in the original environment,
and that forward evolution is possible (the corresponding fitness graph agrees with 1A or 2A).
\begin{conjecture}
Under the assumption stated, there is a correlation between 
that fitness graph 1A represents forward evolution and 
that evolution is reversible.
\end{conjecture}
The conjecture concerns a (possible) statistical correlation, not
a general rule.
Figure 3 and 4 show six two-locus subsystems
that includes $0000$. The conjecture applies
to the two systems defined by the
double mutants $1010$ and $0011$, respectively.
For both systems forward evolution agrees with
Figure 1A, and evolution is reversible.

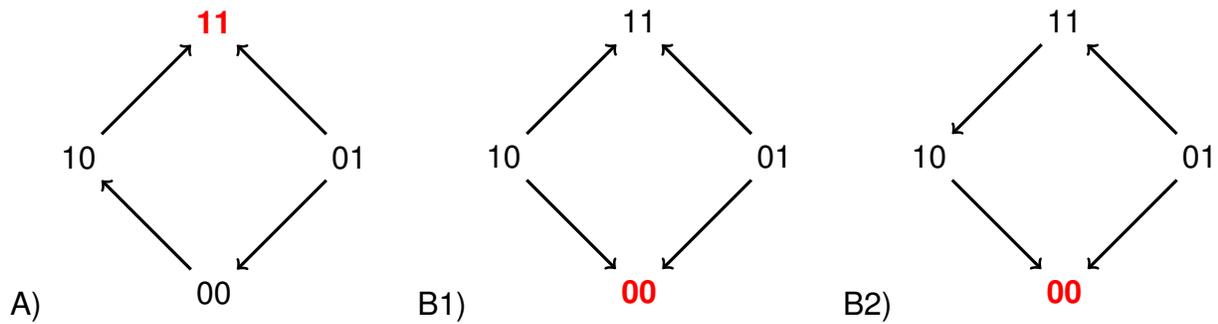
\begin{figure}
A)
\begin{tikzpicture}
[thick, color=black,scale=0.9]
%[scale=0.6, auto=left,every node/.style={circle, draw,
%thick,outer sep=5pt}]
  \node (n1) at (4,0) {00};
  \node  (n2) at (2,2)  {10};
  \node   (n3) at (6,2)  { 01};
  \node [color=red] (n4) at (4,4) {\bf 11};
  \foreach \from/\to in {n2/n4,n3/n4,n1/n2,n3/n1}
  \draw[very thick] [-> ](\from) -- (\to);
\end{tikzpicture}
 \quad
 B1)
\begin{tikzpicture}
[thick, color=black,scale=0.9]
 \node [color=red] (n1) at (4,0) {\bf 00};
  \node  (n2) at (2,2)  { 10};
  \node  (n3) at (6,2)  { 01};
  \node  (n4) at (4,4) {11};
  \foreach \from/\to in {n3/n4,n2/n4,n2/n1,n3/n1}
  \draw[very thick] [-> ](\from) -- (\to);
  \end{tikzpicture}
  \quad
   B2)
\begin{tikzpicture}
[thick, color=black,scale=0.9]
 \node [color=red] (n1) at (4,0) {\bf 00};
  \node  (n2) at (2,2)  { 10};
  \node  (n3) at (6,2)  { 01};
  \node  (n4) at (4,4) {11};
  \foreach \from/\to in {n3/n4,n4/n2,n2/n1,n3/n1}
  \draw[very thick] [-> ](\from) -- (\to);
  \end{tikzpicture}
  
\caption{
The fitness graph for forward evolution (A) has one 
accessible trajectory $00  \mapsto 10 \mapsto 11$.
The genotype $01$ has low fitness in all environments.
Evolution is irreversible if the fitness graph for the original
environment agrees with $B1$,
and reversible if the graph agrees with $B_2$.}
\end{figure}

\subsection{Irreversible evolution and  fluctuating drug concentrations}
We continue with assumptions very similar to Section 2.3, 
except that we consider different drug concentrations.
Specifically $L=2$ and if the drug concentration $C \geq C_T$
for some threshold concentration $C_T$, then $w_{10}>w_{00}$,
otherwise $w_{00}>w_{10}$. For simplicity, we also
assume that $w_{11}>w_{10}$ if $C \geq C_T$, so that the path
$00 \mapsto 10 \mapsto 11$ is accessible as soon as $C \geq C_T$.
The genotype $01$ has low fitness in all environments (similar to the
second case in Sections 2.3).
The relevant graphs are shown in Figure 7,
where 7A represents forward evolution and the two alternatives 
for reverse evolution are 7B1 and  7B2.

By working with very precise assumptions, one can clearly see
some of the mechanisms at play. The prospects for reverse evolution falls naturally into three 
cases described by the tables 1--3. In brief, the possible outcomes are:

\begin{table}
\begin{tabular}{ l  | c | l }
\hline
\hline
Concentration $C$  &  Rank order  &  Peaks  \\
 $C \geq  C_T$  &  $w_{11}>w_{10}>w_{00}$   & 11 \\
\hline
 $C <C_T$  & $w_{00}>w_{10}>w_{11}$ & 00 \\
 \hline
\end{tabular}
\caption
{Evolution is reversable. Regardless of drug concentration, there is one peak in the fitness landscape. The fitness
graph agrees with B2 for all $C<C_T$.}
\end{table}

\begin{table}
\begin{tabular}{ l  | c | l }
Concentration $C$  &  Rank order  &  Peaks  \\
 $C \geq  C_T$  &  $w_{11}>w_{10}>w_{00}$   & 11 \\
 \hline
 $ C_{\tilde{T}}  <  C < C_T$  &  $w_{11}>w_{00}>w_{10}$   & 00 and 11 \\
\hline
 $ C \leq  C_{\tilde{T}}$  &  $w_{00}>w_{11}>w_{10}$   & 00 and 11 \\
  \hline
\end{tabular}
\caption{Evolution is irreversible.
There are two peaks for all concentrations below the threshold $C_T$.
The rank order of $11$ and $00$ changes at some threshold concentration 
$C_{\tilde{T}}< C_T$, but the fitness of $10$ remains low. 
%The fitness graph agrees with $B1$ for all $C< C_T$, even if the rank order
%of $11$ and $00$ switches when $C$ gets sufficiently low (at  $C=C_{\tilde{T}}$
% ).
}
\end{table}

\begin{table}
\begin{tabular}{ l  | c | l }
Concentration $C$  &  Rank order  &  Peaks  \\
 $C \geq  C_T$  &  $w_{11}>w_{10}>w_{00}$   & 11 \\
\hline
 $ \hat{C_T }  <  C < C_T$  &  $w_{11}>w_{00}>w_{10}$   & 00 and 11 \\
  \hline
 $C   \leq \hat{C_T }$  & $w_{00}>w_{10}>w_{11}$ & 00 \\
 \hline
\end{tabular}
\caption{Evolution is reversible. However, the situation is less favorable
than the case described by Table 1, since there is a threshold asymmetry.
The fitness  graph agrees with $B1$ for  $ \hat{C_T }  <  C < C_T$
and with $B2$ for $C \leq \hat{C_T }$, i.e., the threshold concentration for 
development of resistance is higher than the threshold for its reversion.}
\end{table}

\begin{itemize}
\item[(i)] 
evolution can be reversed, and the threshold for reverse evolution is
the same as for forward evolution ($C_T$).
\item[(ii)] evolution is irreversible.
\item[(iii)]  
evolution can be reversed, but the
threshold  for  reverse evolution ($ \hat{C_T}$) is lower than for
forward evolution ($C_T$). 
\end{itemize}
In practical terms, the assymetry in case (iii) means
that resistance can be maintained for drug levels lower than what is necessary for
resistance development. However, reverse evolution is 
possible for a drug-free environment.

The very last case we consider for $L=2$ (Table 4) falls outside the
main topic for the paper because there is no trade-off
between fitness for the different environments (and
consequently no reason to expect reverse evolution). Rather does
the case sorts under cost-free drug resistance.
Similar to Tables 1-3, the wild-type $00$ has maximal fitness in the
drug free environment, $11$ for high concentrations,
whereas $01$ has low fitness in all environment.
However, in contrast to the previous tables,
$11$ has maximal fitness in all environments.

Under the given assumption, suppose that a population is exposed to low
drug concentrations for an extended period of time ($0<C<C_T$).
Then $00$ is a suboptimal peak. 
By assumption, the population cannot reach the global peak $11$ for low concentrations,
unless it is first exposed to high concentrations ($C \geq C_T$).
In other words, the system has "memory" of sort, since exposure to high drug concentrations
causes a permanent change (and increased fitness for any $C>0$).

\begin{table}
\begin{tabular}{ l  | c | l }
Concentration $C$  &  Rank order  &  Peaks  \\
 $C \geq  C_T$  &  $w_{11}>w_{10}>w_{00}$   & 11 \\
\hline
 $ 0  <  C < C_T$  &  $w_{11}>w_{00}>w_{10}$   & 00, 11  \\
  \hline
 $C=0$  & $w_{11}=w_{00}>w_{10}$ & 00, 11 \\
 \hline
\end{tabular}
\caption{Similar to Tables 1--3, the wild-type $00$ has maximal fitness in the
drug free environment, $11$ for high concentrations,
whereas $01$ has low fitness in all environment.
However, in contrast to the previous tables,
there is no longer a trade-off between environments, 
since $11$ has maximal fitness in all environments.}
\end{table}

Given the variation in behavior 
already for $L=2$, it is reasonable to expect interesting
dynamics for larger systems, see \citet{dkm, das}.
Returning to the malaria study, 10 different drug concentrations were considered.
Figure 8 summarizes information for all 10 drugs. Each arrow that 
points up for all 10 concentrations is marked red. The other arrows are black.
The graph shows that evolution from $1111$ to $0000$ is at least theoretically 
possible under fluctuating drug concentrations.

\begin{figure}
\begin{tikzpicture}
[very thick,black,->,outer sep=1mm, scale=0.9]
\node (n0) at (0, -4.0) {0000};
\node (n1) at (4.0, -2.0) {0001};
\node (n2) at (1.3333333333333335, -2.0) {0010};
\node (n3) at (5.0, 0) {0011};
\node (n4) at (-1.333333333333333, -2.0) {0100};
\node (n5) at (3.0, 0) {0101};
\node  (n6) at (1.0, 0) {0110};
\node  (n7) at (4.0, 2.0) {0111};
\node (n8) at (-3.9999999999999996, -2.0) {1000};
\node (n9) at (-1.0, 0) { 1001};
\node  (n10) at (-3.0, 0) {1010};
\node (n11) at (1.3333333333333335, 2.0) {1011};
\node (n12) at (-5.0, 0) {1100};
\node (n13) at (-1.333333333333333, 2.0) {1101};
\node  (n14) at (-3.9999999999999996, 2.0) {1110};
\node (n15) at (0, 4.0) {1111};
\draw (n1) -- (n0);
\draw (n2) -- (n0);
\draw (n4) -- (n0);
\draw (n8) -- (n0);
\draw (n3) -- (n1);
\draw  (n5) -- (n1);
\draw [color=red, line width=2.1 pt] (n1) -- (n9);
\draw (n3) -- (n2);
\draw (n6) -- (n2);
\draw  [color=red, line width=2.1 pt] (n2) -- (n10);
\draw (n5) -- (n4);
\draw (n6) -- (n4);
\draw (n12) -- (n4);
\draw  [color=red, line width=2.1 pt]  (n3) -- (n7);
\draw  [color=red, line width=2.1 pt]  (n5) -- (n7);
\draw (n7) -- (n6);
\draw  (n9) -- (n8);
\draw (n10) -- (n8);
\draw  (n12) -- (n8);
\draw (n11) -- (n3);
\draw (n11) -- (n9);
\draw (n11) -- (n10);
\draw (n13) -- (n5);
\draw (n13) -- (n9);
\draw (n13) -- (n12);
\draw  [color=red, line width=2.1 pt] (n6) -- (n14);
\draw  [color=red, line width=2.1 pt] (n10) -- (n14);
\draw  [color=red, line width=2.1 pt]  (n12) -- (n14);
\draw (n15) -- (n7);
\draw (n15) -- (n11);
\draw (n15) -- (n13);
\draw (n15) -- (n14);
\end{tikzpicture}
\caption{ The graph summarizes information for ten different concentrations of the drugs,
including the drug-free environment.
The red arrows indicate fitness differences that are consistent for all ten concentrations of the drug. In particular,
$1110$ has higher fitness than the double mutants $1100, 1010, 0110$, regardless of concentration.
Each black arrows indicates that fitness increases for at least one concentration of the drug. The graph shows
that fluctuating concentrations could restore the wild-type $0000$.}
\end{figure}
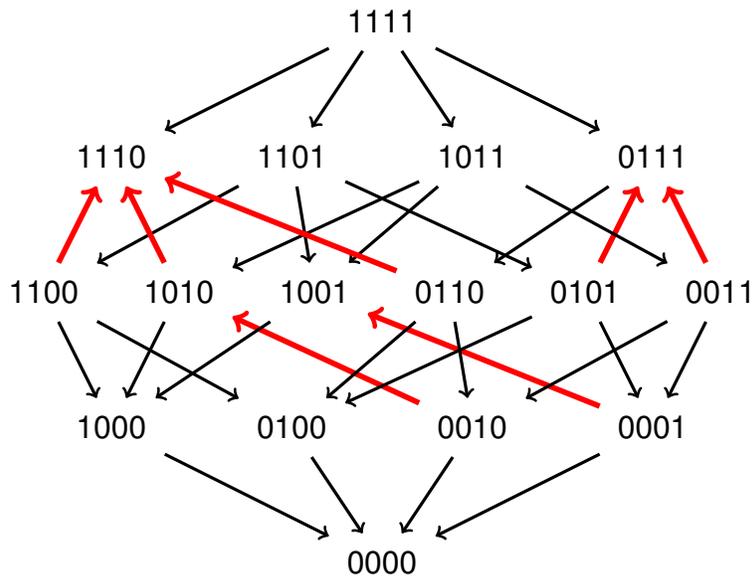

\section{Discussion}
Persistent drug resistance is a multifaceted problem.
Even if resistance is costly, a drug pause
does not necessarily restore susceptibility.
A complete analysis requires consideration of
both evolutionary processes, including reverse
mutations and accumulations of new compensatory
mutations,  and of non-evolutionary mechanisms
such as the  potential for regrowth of the former wild-type
and properties of replacement drugs.
One of the fundamental questions is
if evolution is reversible in principle (regardless
if reversions are plausible or not).
The question of genetic reversibility
is immediately related to obstacles
in fitness landscapes.

For analyzing local and global obstacles
we introduced some new concepts based on rank orders.
A rank order preserving block  is a subset of loci (sometimes referred to as a module) 
with the property that the rank order of genotypes that agree at all loci outside
of the block does not depend on background.
A weaker condition is a graph preserving block, where the
rank order of mutational neighbors that agree at all loci outside
of the block does not depend on background.
The condition implies that the fitness graphs are
similar regardless of background (Figure 5).
The existence of a rank order preserving block
with suboptimal peaks implies that there are also suboptimal
peaks in the global fitness landscape, and likewise for graph preserving blocks.
For a study of irreversible malarial drug resistance \citep{oh}, we identified a double peaked
graph preserving block (modulo a single deviating genotype).

If the $L$-sequence (the genom) can be partitioned into rank order preserving blocks,
the result can be considered a generalization of block landscapes \citep{pm}.
All else equal, adding more blocks does not make the fitness landscape 
more favorable.

In general, rank order induced (or signed) interactions  \citep{c, clg, cgg}
including signed versions 
of higher order epistasis and circuits (introduced to biology
in \citet{bps}),
have been used for analyzing accessibility and obstacles
in fitness landscapes, as well as for
detecting interactions from incomplete data.
Rank order and graph preserving blocks
provide similar insights,
and obvioulsy all the signed concepts are analogous to
sign epistasis \citep{wwc} in that they capture
order implications and are blind for magnitude differences
that do no imact on rank orders.

We considered the relation between fitness landscapes for forward
and reverse evolution.
For $L=2$ we conjecture that absence of sign epistasis for forward
evolution correlate with reversible evolution.
More generally, one can ask if favorable landscapes 
for forward evolution correlate with reversibility.
Results in \cite{dkm, das} are compatible with such a claim,
but more empirical studies would be necessary for a 
conclusion.

For landscapes defined by different drug concentrations, 
it is of interest to compare concentration thresholds for
forward and reverse evolution.
The dynamics for $L=2$ is already interesting. 
We demonstrated that successful adaptation
to low drug concentration may require a history 
of adaptation to high drug concentrations.
We have argued that thresholds asymmetries
are plausible. 
The impact of fluctuating drug concentrations
was analyzed for the study on
irreversible malarial drug resistance
mentioned. Reversion to the original
wild-type was at least theoretically possible, which illustrates
that fluctuating concentration within the range of
two extremes (here high drug concentration and the drug-free environment)
can result in qualitatively different outcomes as compared to 
switches between the extremes.

We have pointed out that peak and rank order preserving 
blocks are incompatible with some standard constructions of
fitness landscapes that use random fitness
(see the discussion about hop-to-top landscapes),
and  that neither redundancy nor decreased peak density by $L$
imply that fitness landscapes are favorable for large $L$.
It appears that no simple summary
statistics can predict whether or not a fitness landscape
is favorable.
A natural category of rugged landscapes 
with good peak accessibility is identified in \cite{das}, which is another indication
that ruggedness alone does  not reveal the character of
a fitness landscape.

However, evolutionary reversibility is a potential indicator of 
fundamental properties of fitness landscapes.
Whether or not the wild-type can be restored
and properties of new genotypes that result from
compensatory mutations carry information about
peak constellations, accessibility and constraints in
the landscape. Sufficiently complete and precise 
empirical studies of resistance and its reversal could 
contribute to a better understanding of microbial evolution,
in particular of microbial fitness landscapes.

\end{document}